\documentclass[english,twoside,11pt]{article}

\usepackage[lmargin=1in,rmargin=1in,tmargin=1in,bmargin=1in]{geometry}

\usepackage{float}
\usepackage{graphicx}
\usepackage{color}
\usepackage{tikz}
\usepackage{verbatim}
\usetikzlibrary{backgrounds}

\usepackage[margin=20pt]{caption}

\usepackage{amsmath}
\usepackage{amssymb}
\usepackage{amsthm}
\usepackage{mathrsfs}

\usepackage[english]{babel}
\usepackage{flafter}
\usepackage{array}
\usepackage{paralist}

\graphicspath{{.}{graphics/}}

\newtheorem{theorem}{Theorem}
\newtheorem{corollary}[theorem]{Corollary}
\newtheorem*{corollary*}{Corollary}
\newtheorem{lemma}[theorem]{Lemma}

\newtheorem*{proposition*}{Proposition}
\newtheorem{property}[theorem]{Property}
\newtheorem*{property*}{Property}

\newtheorem*{observation*}{Observation}

\theoremstyle{definition}
\newtheorem{definition}[theorem]{Definition}
\newtheorem*{definition*}{Definition}

\theoremstyle{remark}
\newtheorem{remark}[theorem]{Remark}
\newtheorem*{remark*}{Remark}

\newtheorem*{example*}{Example}

\usepackage{framed}
\usepackage{algorithm}
\usepackage[noend]{algorithmic}

\newcommand{\la}{\ensuremath{ \mathcal{L}}}
\newcommand{\fl}[1]{\lfloor #1 \rfloor}
\newcommand{\ce}[1]{\lceil #1 \rceil}

\begin{document}

\title{Labeling Schemes for Bounded Degree Graphs}

%

\author{
	David Adjiashvili \\
	\small Institute for Operations Research, ETH Z\"urich \\
	\small R\"amistrasse 101, 8092 Z\"urich, Switzerland \\
	\small email: addavid@ethz.ch
	\and
	Noy Rotbart \\
	\small Department of Computer Science, University Of Copenhagen\\
	\small Universitetsparken 5, 2200 Copenhagen\\
	\small email: noyro@diku.dk
}

\maketitle

\begin{abstract}
We investigate adjacency labeling schemes for graphs of bounded degree $\Delta = O(1)$. In particular, we present an optimal (up to an additive constant) $\log n + O(1)$ adjacency labeling scheme for bounded degree trees. The latter scheme is derived from a labeling scheme for bounded degree outerplanar graphs. Our results complement a similar bound recently obtained for bounded depth trees [Fraigniaud and Korman, SODA 10’], and may provide new insights for closing the long standing gap for adjacency in trees [Alstrup and Rauhe, FOCS 02’]. We also provide improved labeling schemes for bounded degree planar graphs. Finally, we use combinatorial number systems and present an improved adjacency labeling schemes for graphs of bounded degree $\Delta$ with $(e+1)\sqrt{n} < \Delta \leq n/5$.
\end{abstract}

\section{Introduction}\label{sec:Introduction}
	A labeling scheme is a  method of distributing the information about the
	structure of a graph among its vertices by assigning short \emph{labels}, 
	such that a selected function on pairs of vertices can be computed using only their labels.
	The quality of a labeling scheme is mostly measured by its \emph{size}: that is, the maximum number of bits used in a label.
	Additional important attributes of labeling schemes are the running times of the label generation algorithm (\emph{encoder}), and the 
	decoding algorithm (\emph{decoder}), which replies to a query given a pair of labels.

	Among all labeling schemes, that of \emph{adjacency}
	is perhaps the most fundamental, as it directly comprises an implicit representation of the graph. 
	For a graph $G$ and any two of its vertices $u,v$, the decoder of an adjacency labeling scheme is required
	to deduce whether $u$ and $v$ are adjacent in $G$ directly from their labels. Adjacency queries for bounded degree graphs appear naturally in networks of  small dilation~\cite{BCLR}, peer-to-peer (P2P)~\cite{laoutaris2007bounded}  and wireless ad-hoc networks~\cite{wang2006localized}. 

Our main  contribution is an optimal (up to an additive constant) $\log n + O(1)$ size adjacency labeling scheme 
	for bounded-degree outerplanar graphs. As a special case thereof, we obtain an optimal labeling scheme for
	bounded degree trees. We summarize this result in the following theorem.

	\begin{theorem}\label{thm:main}
	For every fixed $\Delta \geq 1$, the class $\mathcal{O}(n,\Delta)$ of bounded-degree outerplanar graphs admits an
	adjacency labeling scheme of size $\log n + O(1)$, with encoding complexity $O(n \log n)$
	and decoding complexity $O(\log \log n)$.
	\end{theorem}

	Our labeling scheme utilizes a novel technique based on \emph{edge-}universal 
	graphs\footnote{A graph $G$ is \emph{edge-universal} for a class $\mathcal{R}$ of graphs, if every
	graph in $\mathcal{R}$ appears as a subgraph in $G$ (not necessarily induced).}
	for bounded degree outerplanar graphs.
	Unlike other results in the field which rely on a tight connection to induced-universal 
	graphs\footnote{A graph $G$ is \emph{induced-universal} for a class $\mathcal{R}$ of graphs, if every
	graph in $\mathcal{R}$ appears as an induced subgraph in $G$.} ~\cite{alstrup2005labeling,alstrup2002small,Gavoille2003115,korman2006constructing}, 
	our technique embeds the input graph into a small edge-universal graph.
	Moreover, to the best of our knowledge, our labeling scheme is the first to use the total label size to separate  the different components of the label.
	In contrast, other labeling schemes, such as~\cite{Thorup01,alstrup02nca}, introduce an extra overhead to support such separation.
		
	Kannan, Naor and Rudich~\cite{kannan1988implicit} showed that if a graph class $\mathcal{R}$ admits
	an adjacency labeling scheme with maximum label length $g(n)$, then there exists
	an induced-universal graph with $2^{g(n)}$ vertices for $\mathcal{R}$, efficiently constructible from the labeling scheme. 
	The opposite relation holds in a weaker sense. The existence of an induced-universal graph with $2^{g(n)}$ 
	vertices for a family $\mathcal{R}$ of graphs implies the existence of labeling scheme with size $g(n)$.
	This transformation is however not efficient, namely the resulting scheme has exponential running time.
	In light of the existing linear size induced-universal graphs for bounded degree trees~\cite{chung1990universal}, our 
	contribution is in devising an \emph{efficient} labeling scheme of optimal size.

	As a corollary of Theorem~\ref{thm:main}, we also obtain an efficient $(\fl{\Delta/2}+1) \log n$ size labeling scheme 
	for graphs with maximum degree $\Delta$.
	For the case of bounded degree planar graphs we construct a $\frac{\ce{\Delta/2} + 1}{2}\log n$ size labeling scheme
	with average label size of $(1 + o(1))\log n + O(\log\log n)$, improving the best known construction for $\Delta \leq 4$.
	Finally, we observe that a simple application of combinatorial number systems~\cite{knuth2011combinatorial}
	gives an adjacency labeling scheme for all graphs with maximum degree $\Delta(n)$,
	improving  the known bounds for $\Delta(n) \in [(e+1)\sqrt n ,n/5]$.

		We summarize all known results for adjacency labeling schemes in Table~\ref{tab:priors}, and our
	contributions in Table~\ref{tab:current}. Our results for outerplanar graphs, planar graphs and
	general graphs are presented in Section~\ref{sec:outerplanar}, Appendix~\ref{apx:planar} and
	Section~\ref{sec:combinadics}, respectively.
	\emph{We defer all proofs to appendix~\ref{apx:proofs}}.

\begin {table}
\caption{Best known adjacency labeling schemes for graphs with at most $n$ vertices.}
\begin{tabular}{ l c l c l c}

\textbf{Family} 			   				 & \textbf{Upper bound}						 & \textbf{Lower bound}	  		& \textbf{Encoding}	  		& \textbf{Decoding}  			& \textbf{Ref.}   \\ \hline
Trees  				 			 &$\log n+O(\log^* n)$ 			 &$\log n+ \Omega(1)$ 		 &$O(n \log^*n)$ 	& $O(\log^* n)$		& \cite{alstrup2002small} \\ 
Binary trees 			 			&$\log n+O(1)$  				&$\log n+ \Omega(1)$		 & $O(n)$	            	&$O(1)$ 			&\cite{bonichon2006short} \\ 
Bd. depth $\delta$	trees 		&$\log n+O(1)$  			&$\log n+ \Omega(1)$		 & $O(n)$	            	&$O(1)$ 			&\cite{fraigniaud2010compact} \\ 
Bd. deg. $\Delta$	trees 		&$\log n+O(\log^* n)$ 			 &$\log n+ \Omega(1)$ 		 &$O(n \log^*n)$ 	& $O(\log^* n)$		& \cite{alstrup2002small} \\ 
Planar graphs 					 	&$2\log n +O(\log \log n)$  		&$ \log n + \Omega(1)$		 & $O(n)$	            	&$O(1)$ 			&\cite{gavoille2007shorter}\\ 
Outerplanar graphs 					 	&$\log n +O(\log \log n)$  		&$ \log n + \Omega(1)$		 & $O(n)$	            	&$O(1)$ 			&\cite{gavoille2007shorter}\\ 
Bd. deg. $\Delta(n)$ graphs	 	&$\lceil \frac{ \Delta(n)}{2} \rceil \log n+O(\log^*n)$  			&$ \frac{ \Delta(n)}{2}  \log n$		 & $O(n)$	            	&$O(\log^*n)$ &\cite{alstrup2002small} \\ 
Graphs  	 						&$\lfloor \frac{1}{2} n \rfloor + \lceil \log n \rceil$  		&$\lfloor \frac{1}{2} n \rfloor -1$ 		 & $O(n)$	         &$O(1)$ 	&\cite{moon1965minimal} \\
\label{tab:priors}
\end{tabular}
\end{table}

\begin {table}
\begin{center}
\caption{Our contribution for families of graphs with bounded degree $\Delta$.}
\begin{tabular}{ l c l c l c}

\textbf{Family} 			   				 & \textbf{Upper bound}						 & \textbf{Tight}	  		& \textbf{Encoding}	  		& \textbf{Decoding}  			& \textbf{Ref.}   \\ \hline
Trees 		&$\log n+ O(1)$  			&yes		 & $O(n \log n)$	  &$O(\log \log n)$ 			&Cor.~\ref{cor:trees}  \\ 
Outerplanar    &$\log n+ O(1) $				&yes		 & $O(n \log n)$	  &$O(\log \log n)$ 			&Thm.~\ref{thm:main}  \\ 
Planar ($\Delta \leq 4$) &$\frac{3}{2}\log n + O(\log \log n)$  &no  & $O(n \log n)$ & $O(\log \log n)$  & Thm.~\ref{thm:zevel} \\ 
Graphs, $\Delta(n)$   & $\log \binom{n}{\lceil \Delta(n) / 2 \rceil}+ 2\log n$  &no&$O(n)$ &$O(\Delta(n) \log n)$ & Thm.~\ref{thm:combina} \\	
Graphs (unbounded)	 	&$\lceil \frac{ \Delta}{2} \rceil \log n+O(1)$  &no		 & $O(n \log n )$	            	&$O(\log \log n)$ 			&Cor.~\ref{cor:graphs} \\ 
\label{tab:current}
\end{tabular}
\end{center}
\end{table}

\subsection{Previous Work}\label{previous_work}
Alstrup and Rauhe~\cite{alstrup2002small} proved that  forests (and trees) in  $\mathcal{G}(n)$ have an adjacency  labeling scheme of $\log n +O(\log^*n)$.\footnote{ $\log^*n$  denotes the iterated logarithm of $n$.} Their technique uses a recursive decomposition of the tree which yields the same $\log n +O(\log^*n)$ label size for bounded degree trees.
	Fraigniaud and Korman~\cite{fraigniaud2010compact} showed  that bounded depth trees have a labeling scheme of size $\log n +O(1)$.
	Bonichon et~al.~\cite{bonichon2006short}  proved that caterpillars and binary trees enjoy a labeling scheme of size $\log(n)+O(1)$ using a method called ``Traversal and Jumping''. 
    In a follow up paper, Bonichon et~al.~\cite{Cyrill07} claimed without proof that the 
    aforementioned methods can be used to achieve the same bound for bounded degree trees. 	
    Chung~\cite{chung1990universal} showed the existence of an induced-universal  graph with $O(n)$ vertices for  bounded degree trees.
    
	Graphs with maximum degree $\Delta$ have \emph{arboricity}\footnote{The \emph{arboricity} of a graph $G$  is the minimum number of edge-disjoint acyclic subgraphs whose union is $G$.} $k(\Delta) = \lceil \Delta  / 2 \rceil $~\cite{kannan1988implicit} thus, by the theorem of
    Nash-Williams~\cite{nash1961edge}, they can be decomposed into $\lceil \Delta  / 2 \rceil $ forests.
    Alstrup and Rauhe~\cite{alstrup2002small} combined this result with their labeling scheme for forests to obtain
    a labeling scheme of size $(\lceil \Delta  / 2 \rceil) \log n  +O(log^* n)$ for bounded degree graphs. 
    They also proved a matching lower bound of $\Omega(k(\Delta) \log n)$.

    Butler  \cite{butler2009induced} constructed an induced-universal graph for  graphs with maximum degree $\Delta$ with $O(n^{ \lceil \frac{\Delta+1}{2}  \rceil} )$ vertices.
	The author notes that any  induced-universal graph must have at least $cn^{\lfloor \Delta/2 \rfloor  +1 }$ vertices for some $c$ depending only on $\Delta$, which implies that the bounds are optimal when $\Delta$ is even.
	For odd $\Delta$, Esperet et~al.~\cite{esperet2008induced} showed  a smaller induced-universal graph with   $O(n^{\lceil \Delta / 2 \rceil -1 / \Delta} \log^{2+2/\Delta}n)$ vertices. 
  It follows that there exists a labeling scheme for $\mathcal{G}(n,\Delta)$ of size 
    ${\frac{\Delta}{2}}\log n$  bits for even $\Delta$, and 
    $ ({\lceil \frac{\Delta}{2} \rceil -1 / \Delta}) \log n  + \log( \log^{2+2/\Delta}n)$ 
    bits for an odd $\Delta$ (but that is not necessary efficient). We summarize the best known bounds
    for adjacency labeling schemes in Table~\ref{tab:priors}.

	\subsection{Preliminaries}
				For two integers $0 \leq k_1 \leq k_2$ we denote $[k_1] = \{1,\cdots, k_1\}$ and $[k_1,k_2] = \{k_1, \cdots, k_2\}$.
				A binary string $x$ is a member of the set $\{ 0,1 \}^*$, and we denote its length by $\vert x \vert$.
				We  denote the concatenation of two binary strings $x,y$  by $x \circ y$.
				
				For a graph $G$ we denote its set of vertices and edges by $V(G)$ and $E(G)$, respectively.
				The family of all graphs is denoted $\mathcal{G}$.
				For any graph family $\mathcal{R}$, let $\mathcal{R}(n) \subseteq \mathcal{R}$ 
				denote the subfamily containing the graphs of at most $n$ vertices.
				The collection of graphs with bounded degree $\Delta$ in
				$\mathcal{R}(n)$ is denoted $\mathcal{R}(n,\Delta)$. 
				The collection of  planar graphs, outerplanar graphs, and  trees,  in $\mathcal{G}(n)$ 
				is denoted $\mathcal{P}(n), \mathcal{O}(n)$ and $\mathcal{T}(n)$, resp. 
				Unless otherwise stated, we assume hereafter $\Delta$ to be constant. 
				To simplify the presentation we suppress all dependencies on $\Delta$ in all our bounds and running time estimations.
				 All these dependencies can be computed and shown to be at most a multiplicative factor of $O(\Delta\log \Delta)$ times the claimed bounds.
				  We defer the exact details to the journal version of the paper.
				Non constant bounds on the degree are denoted by $\Delta(n)$. We note that all 
				results work  for disconnected graphs.  We assume  trees to be rooted, and 
				denote $ \log_2{n} $ as $\log{n}$.
				For a set of vertices $ S \subset V(G)$ we define $G-S$ to be 
				the graph obtained from $G$ by removing the vertices in $S$ and all incident edges.
				The set of edges in $G=(V,E)$ incident to a vertex $v\in V$ is denoted $E_v$.
							    
		Let   $G=(V,E)  \in \mathcal{G}(n)$, and let $u,v \in V$.
	$adjacency(v,u)$ is the boolean function  over vertices in  $G \in \mathcal{G}$ 
	that returns \textbf{true} if and only if $u$ and $v$ are adjacent in $G$.
	A \emph{label assignment}  for $G \in \mathcal{G}$ is a mapping of each $v \in V(G)$ 
	to a bit  string $\la(v)$, called  the \emph{label} of $v$. An  \emph{adjacency labeling scheme} for  $\mathcal{G}$ consists   of  an encoder  and decoder.
	The  \emph{encoder}  is an algorithm that receives $G \in \mathcal{G}$ as input and  computes the label assignment $e_G$.
	The \emph{decoder}  is an algorithm  that receives any two labels $\la(v),\la(u)$  and  computes the query $d(\la(v),\la(u))$, such that  $d(\la(v),\la(u))=adjacency(v,u)$. The \emph{size} of the labeling scheme
 is the maximum label length.
	Hereafter, we refer to adjacency labeling schemes simply as labeling schemes. For the encoding and decoding 
	algorithms, we assume a $\Omega(\log n)$ word size RAM model.

\section{$\log n + O(1)$ Labeling Scheme for Bounded-Degree Outerplanar Graphs}\label{sec:outerplanar}
In this section we describe a labeling scheme for outerplanar graphs with bounded degree $\Delta$.
Our method relies on an embedding technique of Bhatt, Chung, Leighton and Rosenberg~\cite{BCLR}
for bounded degree outerplanar graphs.
 In their paper, the authors were concerned with \emph{edge-universal graphs} for various families of bounded degree graphs. 
 In particular, they show that for every $n \in \mathbb{N}$ there exists a graph $H_n$ with $O(n)$ vertices and $O(n)$ edges that contains every bounded
degree outerplanar graph $G \in \mathcal{O}(n,\Delta)$ as a subgraph (not necessarily induced).
\subsection{Our Methods}
 Our main tool is an embedding technique due to Bhatt et~al.~\cite{BCLR} of outerplanar graphs 
into $H_n$. On the one hand, the embedding is simple to compute. 
This fact will lead to an efficient $O(n \log n)$ time encoder.
 On the other hand, the embedding satisfies a useful locality property.
 This property  allows our labels to contain both unique vertex identifiers of the graph $H_n$ and edge identifiers, without 
exceeding the desired label size $\log n + O(1)$.

To obtain the latter label size via an embedding into $H_n$ we need to overcome several
difficulties. Although $H_n$ has a linear number of edges, its maximum degree is $\Omega(\log n)$,
thus, unique edge identifiers require $\Omega(\log \log n)$ bits, in general.
 Since also $|V(H_n)| = \Omega(n)$, it follows that a label cannot contain an arbitrary combination of vertex identifiers in $V_n$
and edge identifiers at the same time, as it would lead to labels with size $\log n + \Omega(\log \log n)$.
This difficulty is overcome by exploiting the structure of $H_n$ further and constructing unique
vertex identifiers in a particular way that allows reducing the encoding length. This solution
creates an additional difficulty of separating the different parts of the label in the decoding phase. 
This difficulty is overcome by designing careful encoding lengths that minimize the ambiguity, 
and storing an additional constant amount of information to eliminate it altogether.

\subsection{A Compact Edge-Universal Graph for Bounded-Degree Outerplanar Graphs}

	We describe next the edge-universal graph $H_n = (V_n, E_n)$ constructed by Bhatt et~al.~\cite{BCLR}
	for $\mathcal{O}(n,\Delta)$. We let $k = \min\{s \in \mathbb{N} : 2^s-1\geq n\}$ and set $N = 2^k - 1$.
	The construction uses two constants $c = c(\Delta), g = g(\Delta)$, that depend only on $\Delta$.

	The graph $H_n$ is constructed from the complete binary tree $\mathrm{T}$ on $N$ 
	vertices as follows. To obtain the vertex set $V_n$, split every vertex 
	$v\in V(\mathrm{T})$ at level $t$ in $\mathrm{T}$ into $\gamma_t = c \log (N/2^t)$ vertices $w_1(v), \cdots, w_{\gamma_t}(v)$. 
	The latter set of vertices is called the \emph{cluster} of $v$. 
	For $w\in V_n$ we denote by $t(w)$ the \emph{level} of $w$, that is the level in the  binary tree $T$ of the
	cluster to which $w$ belongs.
	
	The edge set $E_n$ is defined as follows. 
	Two vertices $w_i(v), w_j(u)\in V_n$ are adjacent if and only if the clusters they belong to in $\mathrm{T}$ are
	at distance at most $g$ in $\mathrm{T}$.
	Note that $w_1w_2 \in E_n$ implies $|t(w_1) - t(w_2)| \leq g$.
	This completes the construction of the graph. One can easily check that $|V_n| = O(n)$. $|E_n| = O(n)$ also
	holds, but we do not this fact directly.
	The graph $H_n$ is 
	illustrated in Figure~\ref{fig:construction}. Our labeling scheme relies on the following result of Bhatt et~al.~\cite{BCLR}.
	
	\begin{theorem}[{{Bhatt et~al.~\cite{BCLR}}}]\label{thm:embedding}
		$H_n$ is edge-universal for the class of bounded degree outerplanar graphs $\mathcal{O}(n,\Delta)$. 
	\end{theorem}

\begin{figure}[h!]
  \centering
    \includegraphics[width=0.75\textwidth]{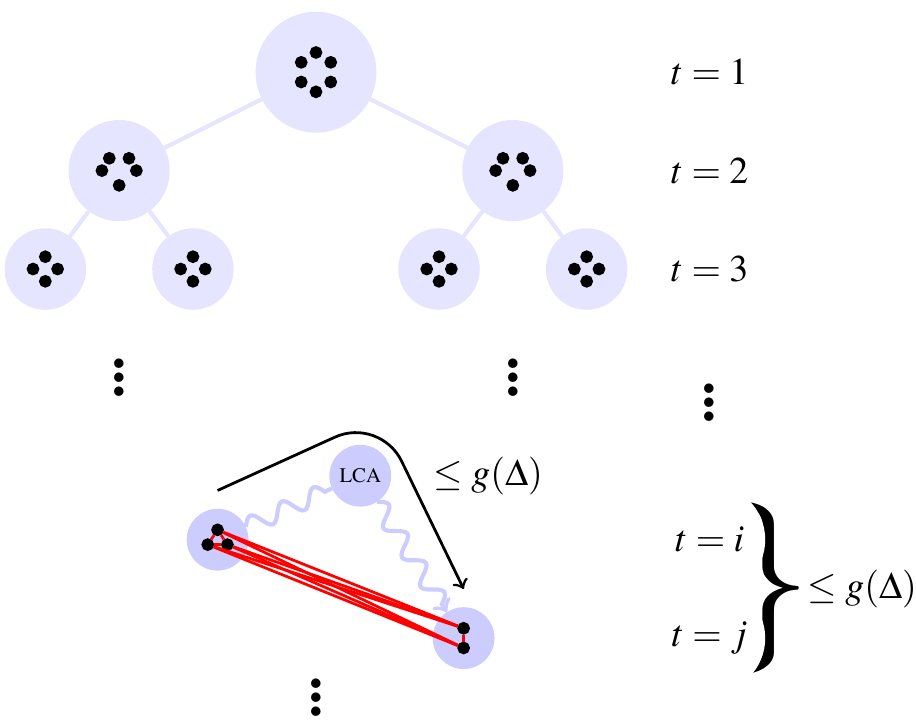}
    \caption{An illustration of the graph $H_n$.}\label{fig:construction}
\end{figure}

	\subsection{Warm-Up: a $\log n +O(\log \log n)$ Labeling Scheme}
	We briefly describe  a simple $\log n + O(\log\log n)$ labeling scheme.
	First, assign unique identifiers $Id$ to the vertices in $H_n$. Since
	$|V_n| = O(n)$ we can assume that $|Id(v)| = \log n + O(1)$ for every $v\in V_n$.
	Next, for every $v\in V_n$ assign unique identifiers to the edges incident to $v$, 
	Since every vertex in $H_n$ has
	$O(\log n)$ neighbours, each edge can be encoded using $\log\log n + O(1)$ bits.
	
	To obtain labels for a given outerplanar graph $G=(V,E)$, compute first an embedding 
	$\phi: V \rightarrow V_n$ of $G$ into $H_n$. Next, define the label of vertex $v\in V$
	to be the concatenation of $Id(\phi(v))$ with the identifiers of all edges
	in $E_n$ leading to images under the embedding $\phi$ of neighbouring vertices, namely all
	$\phi(u)\phi(v)$ such that $uv\in E$. Since the
	maximum degree in $G$ is bounded by the constant $\Delta$, this results in a 
        $\log n +  O(\log\log n)$ label size. It is not difficult to see that encoding and decoding can
	be performed efficiently (we elaborate on it later on).
	
\subsection{The Encoder}
To reduce the size of the labels to $\log n + O(1)$ we need to refine the latter technique significantly.
 As a first step we employ \emph{differential sizing}, a technique first used in the context of labeling 
schemes in~\cite{Thorup01}. In differential sizing some parts of the label do not have a fixed number 
of allocated bits across all labels. Concretely, we use differential sizing for both vertex and edge identifiers.

The resulting labels will have the desired length, but will also contain an undesired ambiguity, 
that will prohibit correct decoding. We will then append a short prefix to the label that will resolve 
this ambiguity.

\subsubsection{Differential Sizing - The Suffix of a Label.}
	 Let us first formally define  our naming scheme for vertex and edge identifiers in $H_n$.

	\begin{definition}\label{def:coherent}
	A \emph{naming} of $H_n$ is an injective function $\mathit{Id}:V_n \rightarrow \mathbb{N}$ and a collection
	of injective functions $\mathit{EId_v}:E_v \rightarrow \mathbb{N}$ for every $v\in V_n$. 
	A naming is \emph{coherent} if for every $v, v_1, v_2 \in V_n$
	\begin{enumerate}
	\item $\mathit{Id}(v_1) > \mathit{Id}(v_2)$ implies that $t(v_1) > t(v_2)$, or $t(v_1) = t(v_2)$ and the cluster
	of $v_2$ appears to the left of the cluster of $v_1$ in $\mathrm{T}$.
	\item $\mathit{EId_v}(vv_1) > \mathit{EId_v}(vv_2)$ implies that $\mathit{Id}(v_1) > \mathit{Id}(v_2)$. 
	\end{enumerate}
	\end{definition}

	We compute a coherent naming by assigning the identifiers $1$ through $|V_n|$ 
	to $V_n$ level by level, traversing the clusters in any single level in 
	$\mathrm{T}$ from left  to right, and then naming the edges incident to 
	$v\in V_n$ from $1$ to $|E_v|$ in a way that is consistent with the vertex naming.

	For $v,v' \in V$ define $\alpha(v) := \ce{\log \mathit{Id(v)}}$ and $\beta(v) := \ce{\log \mathit{EId_v}(vv')}$ and let 
	$$\alpha_t = \max_{v \,:\, t(v) \leq t} \alpha(v) \,\,\,\,\,\,\, \text{and} \,\,\,\,\,\,\,
	\beta_t = \max_{vv' \,:\, t(v) \leq t} \beta(v')$$
	be the maximal number of bits required to encode a vertex name and an edge name for vertices with
	level at most $t$.
	The simple labels described in the beginning of this section store $\log n + O(1)$ and $\log\log n + O(1)$ bits
	for \emph{every} vertex name, and every edge name, respectively. In contrast, our label for a
	vertex in level $t$ stores $\alpha_t$ bits for a vertex name, and $\beta_t$ 
	for edge names. In the following lemma we prove that new labels have the desired size. 

	\begin{lemma}\label{lem:bound}
	For every $t\leq \log N$ it holds that $\alpha_t \leq t + \ce{\log(\log N - t)} + O(1)$,
	$\beta_t \leq \ce{\log(\log N - t)} + O(1)$ and 
	$\alpha_t  + \Delta \beta_t = \log n + O(1)$.
	\end{lemma}
	
	We henceforth denote the part of the label containing the vertex name and all its edge identifiers as the \emph{suffix}.
\subsubsection{Resolving Ambiguity.}
	Since the vertex name does not occupy a fixed number of bits across all
	labels, it is a priori unclear which part of the label contains  it. To resolve this
	ambiguity we analyze the following function, which represents the final length of our
	labels for vertices in level $t$ (up to a fixed constant). Let $D = [\ce{\log N}]$ and
	$f: D\rightarrow \mathbb{N}$ be defined as

	$$
	f(t) = \alpha_t + \Delta \beta_t = t + (\Delta+1)\ce{\log(\log N - t))}.
	$$
	The following lemma states that all but a constant number (depending on $\Delta$) of 
	values in $D$ have at most  two pre-images under $f$. This observation is useful, since it implies that the
	knowledge of the level $t(v)$ of the vertex $v$ can resolve all remaining ambiguities in its label, as the 
	vertex name occupies exactly $\alpha_{t(v)}$ bits.

	\begin{lemma}\label{lem:properties}
		Let $r(\Delta) = \ce{8(\Delta+1)\log(\Delta+1)}$. For every $t \in [\ce{\log N} - r(\Delta)]$ the number of integers  $t'\in D\setminus \{t\}$
		that satisfy $f(t) = f(t')$, is at most one.
	\end{lemma}
	\begin{remark}\label{rem:lemma3}
		It is natural to ask if having equal label lengths for vertices in different levels
		can be avoided altogether (thus making Lemma~\ref{lem:properties} unnecessary). This seems
		not to be the case for the following reason. The number of vertices in every
		level is at least $\Omega(\log N)$, thus, with label size $\ell = o(\log \log N)$ one can not
		uniquely represent all vertices in \emph{any} level. Furthermore, the label length is also 
		restricted to $\log n + O(1)$, and the number of levels is $\ce{\log N}$. Thus, a function assigning
		levels to label lengths would need to have domain $[\ce{\log N}]$ and range $[g(N), \ce{\log N}]$ for
		$g(N) = \Omega(\log\log N)$, implying that it cannot be one-to-one.
	\end{remark}

	Recall that the length of the suffix of  vertex $v\in V$ is exactly
	$\alpha_{t(v)} + \Delta \beta_{t(v)}$.
	We next show how the structural property proved here allows to construct a 
	constant size \emph{prefix}, that will eliminate the ambiguity caused by differential sizing.

\subsubsection{Constructing the Prefix.}
 For a formal description of the prefix we need the following definition. 
 We let $r(\Delta) = \ce{8(\Delta+1)\log(\Delta+1)}$, as in Lemma~\ref{lem:properties}.
\begin{definition}\label{def:label_def}
A vertex $v\in V$ is called \emph{shallow} if its level $t(v)$ is at most 
$\ce{\log N} - r(\Delta)$. We call a shallow vertex \emph{early}
if $t(v)$ is the smallest pre-image of $f(t(v))$. A shallow vertex that 
is not early is called \emph{late}.

A vertex $v\in V$ that is not shallow is called \emph{deep}. A deep vertex is  of 
\emph{type} $\tau$, if its level satisfies $t(v) = \ce{\log N} - \tau$. 
\end{definition}

It is easy to verify the following properties.
Lemma~\ref{lem:properties} guarantees that if $v$ is shallow, then
$f(t(v))$ has at most two pre-images under $f$. If $v$ is shallow and
there is only one pre-image for $f(t(v))$, then $v$ is early. 
Finally, observe that the type of deep vertices ranges in the interval $[1, r(\Delta)]$. 

We are now ready to define the prefix of a label $\la(v)$ for a vertex  $v\in V$. 
Every prefix starts with a single bit $D(v)$ that is set to $0$  if $v$ is shallow, and to $1$
if $v$ is deep. The second bit $R(v)$ in every prefix indicates whether a shallow vertex
is early, in which case it is set to $0$, or late, in which case it is set to $1$. The
bit $R(v)$ is always set to $0$ in labels of deep vertices.
The next part $Type(v)$ of the prefix contains $\ce{\log r(\Delta)}$ bits representing the type
of the vertex $v$, in case $v$ is deep. If $v$ is shallow this field is set to zero.
This concludes the definition of the prefix. Observe that the prefix contains $O(\log \Delta) = O(1)$
bits. We stress that length $s_p$ of the prefix is identical 
across all labels.

It is evident that the prefix of a label eliminates any remaining ambiguity. This follows from
the fact that the level $t(v)$ of a vertex $v \in V$ can be computed from the length of the suffix and the
additional information stored in the prefix. The level, in turn, allows to decompose the 
suffix into the vertex name and the incident edge names, which can then be used for decoding.
We elaborate on the decoding algorithm later on.

\subsubsection{The Final Labels.}

The final label is obtained by concatenating the suffix to the prefix, namely $\la(v)$
is defined as follows.

\[
  \la(v)  = {
    \underbrace{D(v) \circ R(v) \circ Type(v)  }_\text{prefix} \circ
    \underbrace{\mathit{Id}(\phi(v)) \circ \mathit{EId}_{\phi(v)}(e_1) \circ \cdots \circ \mathit{EId}_{\phi(v)}(e_{\Delta})}_\text{suffix}.
   }
 \]

Figure~\ref{fig:labels} illustrates the label structure as a function of the level of the
vertex.
Note that $\vert \la (v) \vert =  s_p + f(t(v)) $, thus  the level $t(v)$ of $v$ 
determines $\vert \la(v) \vert $. Note that Lemma~\ref{lem:bound} and the fact that the prefix has constant
size guarantees that $|\la(v)| = \log n + O(1)$, as desired.
We also  pad each label with sufficiently many $0$'s and a single '$1$', to arrive at a uniform length.
The latter simple modification allows the decoder to  work without knowing $n$ in advance (see~\cite{Fraigniaud10} for details).

      \begin{figure}[h!]\label{fig:labels}
	\centering
	  \includegraphics[width=\textwidth]{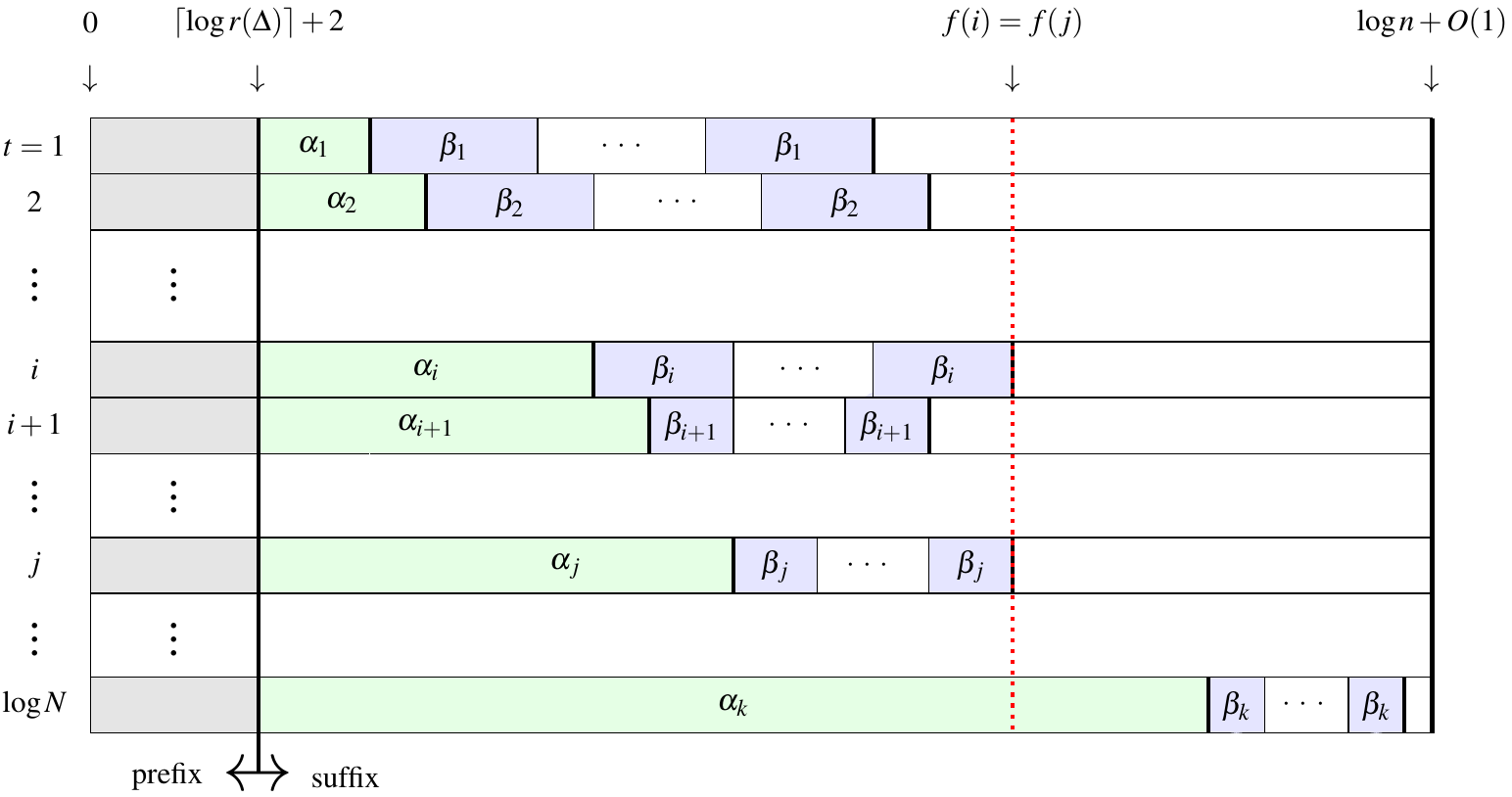}
      \caption{The composition of labels in our labeling scheme for vertices in different levels $t = 1,\cdots, \log N$. 
      The size of the prefix is seen to be constant in every level, while the fields of lengths $\alpha_t$
      and $\beta_t$, used to store vertex and edge identifiers, respectively, have variable sizes. The
      levels $i$ and $j$ comprise a collision with respect to the function $f$, thus labels of vertices in
      these levels have the same length.}
      \end{figure}

Although it is not necessary for the correctness of our labeling scheme, we prove here
uniqueness of the labels. In other words, we show that two different vertices in $G$ necessarily
get different labels.

\begin{lemma}\label{lem:uniqueness}
For every two distinct $u,v \in V(G)$ it holds that $\la(u) \neq \la(v)$.
\end{lemma}

\subsection{Decoding}
Consider two labels $\la(u)$ and $\la(v)$ for vertices $u,v \in V$.
The decoder first extracts the levels $t(u)$ and $t(v)$ of $u$ and $v$~respectively,
using the following simple procedure, which we describe for $v$.
If $D(v)= 0$, $v$ is shallow. To this end, the decoder computes all pre-images
of the length of the suffix, $\vert \la(v) \vert - s_p$, under $f$. Recall, that by Lemma~\ref{lem:properties}, the number
of pre-images is at most two. Let $t_1 \leq t_2$ be the computed pre-images. Next, the decoder
inspects the bit $R(v)$. According to the definition of the labels, $t(v) = t_1$ if $R(v) = 0$, and
$t(v) = t_2$, otherwise.
 Consider next the case  $D(v) = 1$, namely that $v$ is deep. In this case, the decoder inspects
$Type(v)$. The level of $v$ is $t(v) = \ce{\log N} - Type(v)$, by definition of a type of a deep vertex.
It is obvious by the definition of the labels that the decoder extracts $t(v)$ correctly.
Assume next that the decoder extracted $t(u)$ and $t(v)$. The decoder can now extract $\mathit{Id}(\phi(u))$
and $\mathit{Id}(\phi(v))$, by inspecting the first $\alpha_{t(u)}$ and $\alpha_{t(v)}$ bits of the suffix
of $\la(u)$ and $\la(v)$, respectively. Next, the decoder checks if $\phi(u)\phi(v) \not\in E_n$, in which
case it reports \textbf{false}.
Finally, if $\phi(u)\phi(v) \in E_n$ the decoder scans all $\Delta$ blocks of $\beta_{t(u)}$ bits each, succeeding $\mathit{Id}(\phi(u))$
in the suffix of $\la(u)$, checking if one of them contains the edge-identifier $\mathit{EId}_{\phi(u)}(\phi(u)\phi(v))$.
If this identifier is found the decoder reports \textbf{true}. Otherwise, it reports \textbf{false}. 
The correctness of the decoding is clear from the label definition and Lemma~\ref{lem:properties}.
\begin{lemma}~\label{lem:decoding}
The decoding of the labels can be performed in time $O(\log \log n)$.
\end{lemma}

\subsection{Computing  the Embedding $\phi$} 
All the labels can clearly be computed from the graph $G$, the embedding $\phi$
and the graph $H_n$ in time $O(n \log n)$. It is also straightforward to compute $H_n$ in $O(n)$
time. 
It remains to discuss how to compute an embedding $\phi$, for which we  provide a high-level overview.
For a detailed description, see Bhatt et~al.~\cite{BCLR}.

 The algorithm uses a subroutine for computing bisectors of a graph.
A \emph{bisector} of a graph $G= (V,E)$ is a set $S\subset V$ of vertices
with the property that the connected components of the graph $G-S$ can be partitioned into two
parts, such that the sum of the vertices in each part is the same, and no edge connects two vertices 
in different parts. 
If $S$ is a bisector in $G$ we say that $S$ \emph{bisects} $V\setminus S$. 

 Given a $k$-coloring  $V = V_1 \cup \cdots \cup V_k$ of $V$
(for some $k\in \mathbb{N}$), one can define a \emph{$k$-bisector} of $G$ as a set $S\subset V$ 
that bisects \emph{every} color class, namely one that bisects $V_i \setminus S$ for all $i\in [k]$.
An important property of outerplanar graphs is that they admit $O(\log n)$ size $k$-bisectors, for every fixed $k$. 

The algorithm works by assigning vertices in the graph $G$ to clusters in $\mathrm{T}$.
The root of $\mathrm{T}$ is assigned up to  $c \log n$ vertices that form a bisector of $G$ with parts  $G_1, G_2$.
In the next iteration, vertices adjacent to vertices in the bisector are given a new color. Next,
two new $2$-bisectors are found, one in each part $G_1,G_2$, and they are assigned to the  corresponding decedents of the root of $\mathrm{T}$.

Let $k(\Delta) = \log \Delta + 1$. 
 In general, the vertices stored at a vertex of $\mathrm{T}$ at level $t$ correspond to a $k(\Delta)$-bisector. 
 The colors of this bisector  correspond to the neighbors  of vertex-sets stored at $k(\Delta)-1$ nearest ancestors of the current vertex in $\mathrm{T}$.
 The last color is reserved to the remaining vertices.
 Also stored in this vertex are all neighbors of the
vertex-set stored in the ancestor of the current vertex at distance exactly $k(\Delta)$,
that were not yet assigned to some other cluster. We refer the reader to~\cite{BCLR} for an
analysis of the sizes of clusters. 

Let $T(n)$ be the running time of the latter algorithm in a graph with $n$ vertices.
$T(n)$ clearly satisfies $T(n) \leq 2 T(n/2) + O(h(n))$, where $h(n)$ is the complexity
of finding an $O(1)$-bisector of $O(\log n)$ size in an $n$-vertex graph. 
For outerplanar graphs the latter can be done in linear time~\cite{chung1989separator,BCLR}, thus the labels
of our labeling scheme can be computed in $O(n \log n)$ time.

\subsection{Improvements and Special Cases}
Several well-known techniques can be easily applied on top of our construction to 
reduce the additive constant in the label size.
 First, since graphs of maximum degree $\Delta$ have arboricity $\fl{\frac{\Delta}{2}}+1$, one can reduce the number of edge identifiers stored in each label to the latter number (see Kannan et~al.~\cite{kannan1988implicit}). 
We later show a simpler procedure that works for bounded degree graphs.

Finally, for bounded-degree trees $\mathcal{T}(n,\Delta)$, it suffices to store a single edge identifier 
(corresponding to the edge connecting a vertex to its parent in $G$).  We summarize this result in the following corollary of Theorem~\ref{thm:main}.

\begin{corollary}\label{cor:trees}
For every fixed $\Delta \geq 1$, the class $\mathcal{T}(n,\Delta)$ of bounded-degree trees admits 
a labeling scheme of size $\log n + O(1)$, with encoding complexity $O(n \log n)$
and decoding complexity $O(\log \log n)$.
\end{corollary}

%

\section{Labeling schemes for $\mathcal{G}(n,\Delta)$ and  $\mathcal{G}(n,\Delta(n))$  }\label{sec:combinadics}
First we note that Theorem~\ref{thm:main} implies almost directly a 
	$\lceil \Delta /2 \rceil \log n$ labeling scheme for graphs of fixed bounded degree $\Delta$.
	The result follows  from the technique of Alstrup and  Rauhe~\cite{alstrup2002small}, Lemma~\ref{lem:uniqueness}   
	and the fact that any subtree of a bounded degree graph has bounded degree.
	\begin{corollary}\label{cor:graphs}
		For every  $\Delta \geq 1$, the class $\mathcal{G}(n,\Delta)$ of bounded-degree graphs admits 
		labeling schemes of size $\lceil \Delta /2 \rceil \log n$, with encoding complexity $O(n \log n)$
		and decoding complexity $O(\log \log n)$.
	\end{corollary}
	
	From here on, we discuss labeling schemes for graphs of non-constant bounded degree $k= \Delta(n)$.
	Adjacency relation between any two vertices may be reported in only one of the labels representing them.
	 For  bounded degree graphs, the method of Kannan et~al.~\cite{kannan1988implicit} of decomposition into forests can be replaced with a simpler procedure, using Eulerian circuits, as we prove in the following.
	\begin{lemma}\label{THM:split}
		Let  $G=(V,E)$ be a graph  with degree bounded by $k$.
		There exist an adjacency labeling scheme for $G$ of size $ (\lceil \frac{k}{2} \rceil+1) \log n$.
	\end{lemma}
	The current best labeling schemes for graphs works in two modes, according to the range of $k$.
	If $k \leq n/ \log n$, a $k/2 \log n $ labeling scheme can be achieved~\cite{kannan1988implicit}, essentially  by encoding an adjacency list.\
	For larger $k$, labels defined through the adjacency matrix of the graph have size $n/2+ \log n$~\cite{moon1965minimal}.
	Our improved labels use the well-known \emph{combinatorial number system} (see e.g.~\cite{knuth2011combinatorial}).	
		\begin{lemma}   \label{lemma:combinadics}
			 Let $L = \binom{n}{k} $.
			There is a bijective mapping $\sigma: S_k\rightarrow [0,L-1]$ between  the set of strictly increasing sequences $S_k$ of the form $0\leq t_1  < t_2  \dots < t_k < n$ and $[0,L-1]$  given  by
			 $$ \sigma(t_1,\cdots,t_k) =  \sum_{i=1}^k \binom{t_i}{i}.$$
		\end{lemma}
	We use Lemma~\ref{lemma:combinadics} to prove the following theorem.	
	For this purpose we assume that the labeling scheme presented in  Lemma~\ref{THM:split} returns a subset of vertices for every vertex $v \in V$ according to the partition instead of a final label.
	\begin{theorem}\label{thm:combina}
	For   $1 \leq k \leq n$, there exist an adjacency labeling scheme  for  $\mathcal{G}(n,k)$ with  size:
	$ \binom{n}{\lceil k / 2 \rceil} + \lceil \log n \rceil + \lceil \log k \rceil $.
	\end{theorem}
	The labeling scheme suggested  in Theorem~\ref{thm:combina}  implies a  label size of approximately $(k+2)\log n$ bits, when $k$ is small and $ \Theta(n)$ when  $k = \Theta(n)$.
	The following lemma  identifies the range of $k$ for which our labeling scheme improves on the best known bounds.
	\begin{lemma}\label{lemma:Mathias}
		For $(e+1)\sqrt{n} \leq k \leq \frac{n}{5}$, and  $ f(n,k) = \binom{n}{\lceil k / 2 \rceil} + \log k + \log n$ 
		it holds that 
		\begin{inparaenum}[\itshape a\upshape)]  \item $f(n,k) < \frac{n}{2}$; and  \item $f(n,k)< \lceil k/2  \rceil+2 \log(n)$. \end{inparaenum}
	\end{lemma}
We conclude from Lemma~\ref{lemma:Mathias} that our labeling scheme is preferable to
\cite{moon1965minimal,butler2009induced} for graphs of 
bounded degree $k$ for $(e+1) \sqrt{n} \leq k \leq \frac{n}{5}$.

%
%
%



\label{Bibliography}
\bibliographystyle{plain}

\appendix

\section{Missing Proofs from Section~\ref{sec:outerplanar}}\label{apx:proofs}

\subsection{Proof of Lemma~\ref{lem:bound}}\label{apx:bound}

	We start with the following simple observation stating a couple of facts
	about the size of different parts of $H_n$. We let $V_n^t \subset V_n$ denote the set of vertices in $H_n$
	with level at most $t$. This properties will provide bounds on $\alpha_t$ and $\beta_t$.
		\begin{property}\label{prop:simplebounds}
			The following properties hold for every level $t$ and every vertex 
			$v$ in level $t$.
			\begin{enumerate}[i.]
				\item $|V_n^t| \leq c2^{t}(\log N - t + 1)$.
				\item $|\Gamma(v)| \leq 6c\Delta^2(\log N - t + g(\Delta))$.
			\end{enumerate}
		\end{property}

	\begin{proof}
	\textit{i.} Note that the number of vertices in level $i$
	is $c2^{i-1}(\log N - i)$, thus we obtain 
	$$
	|V_n^t| = \sum_{i=1}^t 2^{i-1}c(\log N - i) \leq c2^{t}\log N - c\sum_{i=1}^t 2^{i-1}i.
	$$
	Using $\sum_{i=1}^t 2^{i-1}i = t2^{t+1} - (t+1)2^t +1 \geq 2^t(t-1)$ 
	and substituting in the expression above we obtain the desired result.
	
	\textit{ii.} We notice that, by definition of $H_n$, the set of neighbors 
	of $v$ are exactly all vertices whose cluster is at distance at most 
	 $g(\Delta)$ from the cluster of $v$. It follows that every such cluster
	has at most $c\log (N/2^{t-g(\Delta)}) = O(\log N - t + g(\Delta))$ vertices. 
	Finally, notice that the number of clusters that are at distance at most
	$g(\Delta)$ away from a given cluster is at most $3\cdot 2^{g(\Delta)-1} = 6\Delta^2$. Putting things 
	together we obtain a bound of $6\Delta^2(\log N - t + g(\Delta))$, as desired. $\qed$
	\end{proof}


	  
	  We can now use Property~\ref{prop:simplebounds} to prove the bounds
	  $$
	  \alpha_t \leq t + \ce{\log (\log N -t)} + \log c + 2.
	  $$
	  and
	  $$
	  \beta_{t} \leq \ce{\log(\log N-t)} + 2\log \Delta + \log c + 4.
	  $$

	  We turn to proving the bound on $\alpha_t + \Delta \beta_t$ for an arbitrary $t\leq \log N$.
	  We drop the additive terms $\log c + 2$ and $2\log \Delta + \log c + 4$ 
	  as they do not depend on $n$, and hence only contribute a constant additive term.
	  Substituting the bounds above for $\alpha_t$ and $\beta_t$ and 
	  defining $Q = \log N - t(v)$ we obtain
	  $$
	  \alpha_t + \Delta \beta_t = t + (\Delta+1) \log (\log N-t) = \log N - Q + (\Delta+1) \log \log Q.
	  $$
	  Using the fact that $(\Delta +1)\log \log Q - Q \leq 0$ whenever $Q \geq (\Delta  +1)\log (\Delta +1)$ and 
	  $(\Delta  +1)\log \log Q - Q \leq (\Delta  +1)\log \log (\Delta +1)$ whenever $Q < (\Delta +1)\log (\Delta +1)$, 
	  we have $\alpha_t + \Delta \beta_t \leq \log N + (\Delta +1)\log\log (\Delta +1) = \log n + O(1)$.
	  This concludes the proof of the lemma.

\subsection{Proof of Lemma~\ref{lem:properties}}\label{apx:properties}
Consider first the behavior of the term $h(t) = (\Delta+1)\ce{\log(\log N - t)}$ in the definition
	of $f$. Let $t_1 \in [D]$ be the smallest integer such that $\log(\log N - t_1)$ is integer. 
	It follows that $\log N - t_1 = 2^s$ for some integer $s$. Consider the smallest integer $t_2\in D$ with $t_2 > t_1$ 
	such that $\log(\log N - t_2)$ is integer. This value clearly satisfies $\log N - t_2 = 2^{s-1}$.
	The smallest integer $t_3\in D$ with $t_3 > t_2$ such that $\log(\log N - t_3)$ is integer is
	$\log N - t_3 = 2^{s-2}$ and so on. Let $t_1 < t_2 < \cdots < t_p$ be the obtained sequence of
	$p = O(\log \log N)$ points that satisfies $t_i = \log N - 2^{s-i+1}$ for $i\in [p]$.
	We call such points \emph{jumps}. See Figure~\ref{fig:TwoNumbersImage} for an illustration of $f$.
	Observe next that for every $k\in [p-1]$ the term $h(t)$ is constant within the interval 
	$I_k = [t_k +1, t_{k+1}]$, while at the point $t_{k+1}+1$ it decreases by $\Delta+1$. 
	
	We rewrite $f$ as
	$$
	f(t) = t - \sum_{k\,:\, t_k \leq t-1} (\Delta+1) + (\Delta+1)\ce{\log(\log N - 1)},
	$$
	The last term does not depend on $t$, thus it suffices to prove the claim  for $\hat f(t) = t - \sum_{k\,:\, t_k \leq t-1} (\Delta+1)$.
	 Consider next some  $t\in D$ and assume that $\hat f(r) \neq \hat f(t)$ for every $r < t$.
	Define $K = \ce{\log N} - t$, and let $t_i \geq t$ be the smallest jump larger or equal to $t$.
	Define $\gamma = t_i-t$.
	Observe that the number of jumps  within the interval $[t_i, \ce{\log N}]$, whose length is $K - \gamma$ is  at most $\log (K- \gamma) + 1$, while the next jump occurring after $t_i$
	appears exactly in the middle of the latter interval, namely  $t_{i+1} - t_i = \frac{K-\gamma}{2}$.
	To this end assume that $\hat f(t)$ has at least two other pre-images $t'' > t' > t$ under $\hat f$.
	Since $\hat f$ is monotonic between jumps, there must be at least one jump between
	every consecutive pair of pre-images of $\hat f$. It follows that $t \leq t_i \leq t'$ and
	$t' \leq t_j \leq t''$ for some $j \geq i+1$. In particular, it follows that $[t_i,t_{i+1}] \subset [t_i,t'']$,
	which implies $t'' - t_i \geq \frac{K - \gamma}{2}$. Now, using the fact that within the interval
	$[t_i, \ce{\log N}] \supset [t_i, t'']$ there are at most $\log (K-\gamma)$ jumps, we can write
	\begin{eqnarray*}
	\hat f(t'') &=& t'' - \sum_{k\,:\, t_k \leq t''-1} (\Delta+1) \\
		    &\geq& t - \sum_{k\,:\, t_k \leq t-1} (\Delta+1) + (t''-t) - \sum_{k\,:\, t_k\in [t_i, t''-1]} (\Delta+1) \\
		    &\geq&  \hat f(t) + (\gamma + \frac{K-\gamma}{2}) - (\Delta+1) (\log (K - \gamma) + 1) \\
		    &\geq&  \hat f(t) + \frac{K}{2} - (\Delta+1) (\log K + 1),
	\end{eqnarray*}
	which, using $\hat f(t) = \hat f(t'')$, implies 
	$$
	\frac{K}{2} - (\Delta+1) (\log K + 1) \leq 0.
	$$
	Rearranging the terms we obtain $\frac{K}{2(\log K + 1)} \leq \Delta+1$, which clearly implies $K \leq r(\Delta)$
	and $t \in [\ce{\log N} - r(\Delta)]$. This concludes the proof.

		\begin{figure}[!ht] 
				\centering
				\includegraphics[width=100mm]{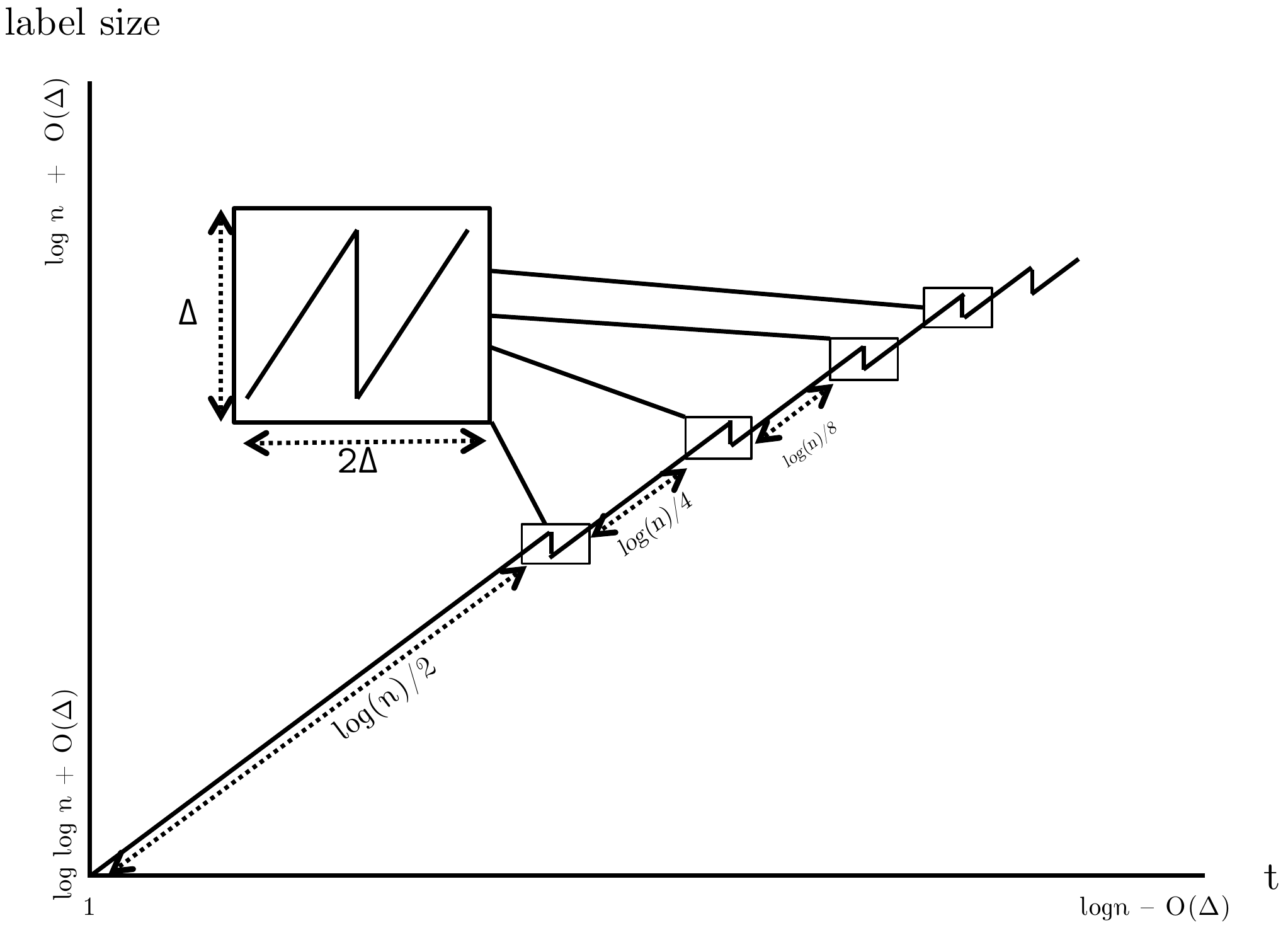}
				\caption{An illustration of Lemma~\ref{lem:properties}}
				\label{fig:TwoNumbersImage}
			\end{figure}

%

\subsection{Proof of Lemma~\ref{lem:uniqueness}}\label{apx:uniqueness}

Consider two vertices $u,v\in V(G)$ and assume $\la(u)=\la(v)$. Since the prefix of every
label has the same length and the same parts for every vertex in $V(G)$, it must
hold that $D(u) = D(v)$, $R(u) = R(v)$ and $Type(u) = Type(v)$.

Assume first that both $u$ and $v$ are shallow. In this case $\la(u) = \la(v)$ implies
that $f(t(u)) = f(t(v))$. Now, since $R(u) = R(v)$, both $v$ and $u$ are either early, 
or late, implying by Lemma~\ref{lem:properties} that $t(u) = t(v) =: t$. It follows by definition 
of the labels, that the first $\alpha_t$ bits of the suffices of $\la(u)$ and $\la(v)$
contain $Id(\phi(u))$ and $Id(\phi(v))$, respectively, implying $Id(\phi(u)) = Id(\phi(v))$.
From the correctness of the embedding $\phi$ and uniqueness of the identifiers $Id$ it follows 
that $u=v$.

Finally, assume that both $u$ and $v$ are deep. In this case $Type(u)=Type(v)$ implies
again that $t(u) = t(v)$. We now reach the conclusion $u=v$ as in the previous case.

\subsection{Proof of Lemma~\ref{lem:decoding}}\label{apx:decoding}

All operations performed by the decoder take $O(1)$ time except for the computation of
pre-images of an integer under $f$, and steps that inspect properties in $H_n$. The former
can obviously be performed in $O(\log \log n)$ time as follows. If the label corresponds
to a deep label, then the decoder checks all $r(\Delta) = O(1)$ possible preimages, so assume next
that the vertex is shallow. Notice that $t - f(t) = O(\log \log n)$ for every $t\in D$. 
Thus, given $f(t)$, one can compute in $O(\log \log n)$ time all $f(y)$ for all integers $y$ with 
$|y - f(t)| = O(\log\log n)$. All preimages of $f(t)$ under $f$ will be found this way.  

We focus hereafter on the complexity of deciding if $\phi(u)\phi(v) \in E_n$ and the 
computation of the edge-identifier corresponding to two vertices in $H_n$. 
To this end we need to elaborate on the way labeling is performed for the graph $H_n$.
Recall that $H_n$ is built from the complete binary tree $\mathrm{T}$ with $\log N$ levels by
splitting each vertex at level $t$ into $\gamma_t =  c \log (N/2^t) $ vertices, thus forming clusters, and then 
connecting two vertices if they belong to clusters at distance at most $g(\Delta) = 2 \log \Delta + 2$
apart in $\mathrm{T}$. 
The labeling of vertices is performed by assigning a \emph{range} $R_C$ of size
$\gamma_t$ to every cluster at level $t$ in a single breadth-first search traversal, and then 
ordering the vertices in the cluster arbitrarily using unique values in the range $R_c$. 
More precisely, the ranges are constructed 
starting from the root of $\mathrm{T}$, processing $\mathrm{T}$ level by level as follows. The Range of the
root cluster is simply $[1,\gamma_1]$. The range of the left child of the root is assigned the 
range $[\gamma_1+1, \gamma_1 + \gamma_2]$, while the right child is assigned the range
$[\gamma_1 + \gamma_2 + 1, \gamma_1 + 2 \gamma_2]$ and so on. At level $t$, the assignment
of ranges starts from the left descendant of the parent at level $t-1$ with the lowest range
(the range with the smallest lower bound),
followed by the right one, and so on. This process defines an ordering of the clusters,
in each level. We denote by $POS(C)$ the position of the cluster $C$ in this order.  

Given a label $L$ of a vertex $w\in V_n$ at level $t = t(w)$, its clusters range (which identifies the cluster) can be computed as follows.
By definition of the labeling, we have that $L$ is contained in the range $I(w) = [A+1, A+ \gamma_t]$ with
$$
A = \sum_{i=1}^{t-1} 2^{i-1} \gamma_i + (POS(C(w)) - 1)\gamma_t,
$$
where $C(w)$ is the cluster to which $w$ belongs.
We now notice that 
$$
\sum_{i=1}^{t-1} 2^{i-1} \gamma_i = \sum_{i=1}^{t-1} 2^{i-1}(\log n - i) = 
c\log n \sum_{i=1}^{t-1} 2^{i-1} - \frac{c}{2} \sum_{i=1}^{t-1} 2^i i,
$$
where both sums in the last formula have simple closed form expressions. It follows
that $POS(C(w))$ can be easily computed from $t$ and $L$ in $O(1)$ time.

Finally, the decoder can compute both the parent and the descendants of a cluster $C$ at 
level $t$ whose range is $[A+1,A+\gamma_t]$. The parent is computed using the fact that
it is at level $t-1$ and its position in this level is $\ce{POS(C)/2}$. The descendants
can be computed using the fact that their level is $t+1$ and their positions in this level
are $2POS(C)-1$ and $2POS(C)$. Both computations can clearly be performed in $O(1)$
time, once presented with the values $\gamma_t$.
We conclude that it is possible to obtain in time $O(1)$ the ranges of all neighbours 
of a given cluster $C$ in $\mathrm{T}$. 

We turn to the problem of deciding whether $\phi(u)\phi(v) \in E_n$. The decoder starts
by computing the clusters $C_u$ and $C_v$ of $u$ and $v$, respectively in time $O(\log n\log n)$.
Then the decoder computes paths $P_u$ and $P_v$ of length at most $g(\Delta)$ from $C_u$ and $C_v$,
respectively in the tree $\mathrm{T}$. Each path is obtained by successively moving at most $g(\Delta)$ steps from a 
cluster to its parent in $\mathrm{T}$.
The computation of these paths takes $O(1)$ time, using the aforementioned method for 
computing a parent of a cluster.
By definition of $H_n$, we have that $\phi(u)\phi(v) \in E_n$ if and only if
the clusters of $\phi(u)$ and $\phi(v)$ are at distance at most $g(\Delta)$ from one another in $\mathrm{T}$.
This can be now easily tested by inspecting the paths $P_u, P_v$, which must contain the
least common ancestor (LCA) of $\phi(u)$ and $\phi(v)$, in case they are neighbours in $H_n$ (see Fig.~\ref{fig:construction}
for an illustration).

With a similar argument one shows that each edge identifier can be decoded in $O(1)$ time.
The details are similar, and thus omitted. Performing the decoding to all $\Delta$ edge identifiers
gives a total running time of $O(1)$.
We conclude that the decoder can be implemented to have running time $O(\log \log n)$. 


\subsection{Proof of Lemma~\ref{THM:split}}\label{apx:split}
		We first create an Eulerian multigraph $G'=(V,E')$ from $G$ by adding at most $|V|/2$ new
		edges, comprising the a matching that connects pairs of vertices in $G$ with odd degree.
		The degrees of all vertices in $G'$ are even, hence it is Eulerian.
		We proceed by finding  an Eulerian  circuit $P$ in $G'$ and directing every edge according to the direction along $P$.
		 Every vertex in $G'$  with degree $d$ has now exactly $ \frac{d}{2} $ incoming and outgoing edges.
		The label of a vertex will only correspond to the vertices which are adjacent through  outgoing edges in $P$. 
		This number is at most $ \frac{d}{2} $. The identifiers of the edges in  $M$ are not included in the resulting label.
		The labeling scheme obtained  will  assign unique labels to each of the vertices, and concatenate the selected labels to each vertex.
		since $d \leq k$ this yields a label of size $ (\lceil \frac{k}{2} \rceil+1) \log n$.	

  \subsection{Proof of Theorem~\ref{thm:combina}} \label{apx:combina}
		We assume  the vertices of the graph $G=(V,E)$ to be numbered from $0$ to $n-1$.
					We call a binary vector $C$ of length $n$  an \emph{adjacency vector} of a vertex $v \in V$ if it satisfies that $C_i=1$ if and only if $v_i$ is a neighbor of $v$ in $G$.
					We denote by $\bar C$ the binary vector with $\bar C_i =1 $ if and only if $ C_i = 0$.
					We interpret the vectors $C$ and $\bar C$ and sets $S\subset V$ as sequences of integers in the range $ [0,n-1]$,  as in Lemma~\ref{lemma:combinadics}.
	 				Let $v \in V$ be a vertex with $k' \leq k$ neighbors, and let  $C$ be its adjacency vector.
					Our labeling scheme assigns  every $v \in V$ to an appropriate $\la(v)$ as follows.
					
					We first use the encoder  from the labeling scheme described in Lemma~\ref{THM:split} to obtain a temporary subset $S_v$ of neighbors of $v$.
					If $k'< \frac{2n}{3}$   we  set the last bit of $\la(v)$ to 0, and append it to the number mapped to the sequence of integers corresponding to $S_v$ under the bijection in Lemma~\ref{lemma:combinadics}. 
					 According to Lemma~\ref{THM:split},  we are  assured  that $\vert S_v \vert \leq \ce{\frac{k}{2}} \leq \ce{\frac{n}{3}}$.
					If $k' \geq \frac{2n}{3}$, we set the last bit of $\la(v)$ to $1$, and append to it the number corresponding to  $\bar C$ under the bijection in Lemma~\ref{lemma:combinadics}. 
					Since $k'>\frac{2n}{3}$, the number of $1$'s in $\bar C$ is at most $\frac{n}{3}$.
					Finally, we attach to every label a binary representation of $k'$ and a unique vertex  identifier using exactly $\log n$ bits each.
					The encoder performs the operation in polynomial time. It is straightforward to verify the claimed label length.		
	
					The decoder receives  $\la(u)$ and $\la(v)$  and extracts the corresponding vectors $C_u$ and $C_v$,
					the adjacency vectors of $u$ and $v$, resp., using Lemma~\ref{lemma:combinadics}, and possibly  a bit inversion operation.
					The decoder returns \textbf{true} if and only if  either $C_u(v)=1$ or $C_v(u)=1$.
					We note that the decoding can be implemented  in $O( k \log \binom{n}{k} )$   time.

\subsection{Proof of Lemma~\ref{lemma:Mathias}}\label{apx:Mathias}
Stirling's approximation  yields the following asymptotic approximation.	
	$$ \log \binom{cn}{n} \sim \log \left(\frac{c^c}{{(c-1)}^{c-1}} \right)\cdot n.$$
	Accordingly, $ \log \binom{n}{\frac{1}{3}n} \sim \frac{2.75}{3}n$. Since   $  \log ( \frac{ 10^{10} }{9^{9}})/10 < 0.5 $, and since the function  $\binom{n}{k}$ is increasing for  $\sqrt{n} \leq k \leq \frac{n}{3}$ our labeling scheme will incur strictly less than $0.5n$ label size for the proposed labeling scheme over the range specified.
	
Let $x_n$ be defined by $x_n = \log (n!) - (n \log n - n + 1/2(\log 2 \pi n))$, and by its definition $x_n \rightarrow 0$ as $ n \rightarrow \infty$. In addition:
$$ \log \binom{k}{n} = \log(n!) - \log ((n-k)!) - \log (k!)$$
$$ = (x_n- x_{n-k}-x_k) + n \log n - (n -k)\log(n -k) -k (\log k) +1/2 \log (n/ k(n-k)2 \pi)). $$
Define $f(n,k)$ by $f(n,k) = k/2 \log n - \log \binom{k}{n}$. We are now interested in the smallest $k$ for which $k_n = f(n,k) \leq 0$ for fixed $n$. 
$f(n, \lfloor \sqrt n \rfloor) < 0$, so $k_n > \sqrt n $, so we  can assume $\sqrt n \leq k \leq n/2$, which implies that  $ 1/2 \log (n/(k(n-k)2\pi)) = O(\log n)$.
Furthermore $x_n- x_{n-k}-x_k =O(1)$ so
$$ f(n,k) = k/2 \log n - n \log n + (n-k) \log (n-k) + k \log k +O(\log n)$$
$$ = k/2 \log n + n \log ((n-k) / n )+ k \log (k / (n-k)) + O( \log n).$$

Note that as $ n \log((n - k)/n) = O(k)$ this means that

$$ \frac{f(n,k)}{k} = \frac{1}{2}\log n - \log ((n-k)/k)+O(1) = \log \left( \frac{k \sqrt n}{n-k} \right)  +O(1) $$

This means that there exists some constant $c > 0$ such that if  $\log \left( \frac{k \sqrt n}{n-k} \right) \geq c$
 then $ f(n,k)  \geq 0 $. But

$$ \log \left( \frac{k \sqrt n}{n-k} \right)  \geq c \Longleftrightarrow k \geq c \frac{n}{\sqrt{n} +c }.$$

If  $k > c \sqrt n$ then $f(n,k) \geq 0$. We can conclude $ k_n \leq c \sqrt n$, and we may assume that $ c\sqrt{n} \geq k \geq \sqrt{n}$.

Now we can conclude that 
$n log((n - k)/n) = -k + O(k^2/n) = -k + O(1)$. Hence

$$\frac{f(n,k)}{k} = \frac{1}{2} \log n - 1 + k \log (k / (n- k))+ O ( \frac{\log n}{\sqrt n}),$$
and thus
$$\frac{f(n,k)}{k} = \log \left( \frac{k \sqrt n }{e(n-k)} \right) +O( \frac{\log n}{\sqrt n}).$$
We can now conclude with $c_n = exp( O(\frac{\log n}{\sqrt n}))$
$$k_n = c_n e \frac{n}{\sqrt{n} + c_n e}+ O(1) = e \sqrt n + O(\log n).$$

\section{Labeling Schemes for Bounded-degree Planar Graphs}\label{apx:planar}

We present here our labeling scheme for the class of bounded degree planar graphs $\mathcal{P}(n,\Delta)$.
Again, we rely on
an embedding of the given planar graph $G$ into a graph $PL_n = (W_n, A_n)$, obtained from the complete binary
tree $\mathrm{T}$. The only difference between $PL_n$ and $H_n$ is that in $PL_n$, every vertex $v$ in level
$t$ of $\mathrm{T}$ is divided into a cluster of $\delta_t = c' \sqrt{\frac{N}{2^{-t}}}$ vertices 
$z_1(v), \cdots, z_{\delta_t}(v)$ (instead of $\gamma_t$ in $H_n$), for some constant $c'$. 
As for $H_n$, two vertices in $W_n$ are connected if and only if the distance between their clusters in $\mathrm{T}$ is at most
$\hat g(\Delta) = O(\log \Delta)$. Bhatt et~al.~\cite{BCLR} showed the following.

\begin{theorem}[{{Bhatt et~al.~\cite{BCLR}}}]\label{thm:embedding_planar}
$PL_n$ is edge-universal for the class of bounded degree planar graphs $\mathcal{P}(n,\Delta)$. 
Furthermore, $PL_n$ has $O(n)$ vertices and $O(n\log n )$ edges.
\end{theorem}

We use the latter theorem to prove the following result, that is proved in Appendix~\ref{apx:planar}.

\begin{theorem}\label{thm:zevel}
For every fixed $\Delta \geq 1$, the class $\mathcal{P}(n,\Delta)$ of bounded-degree planar graphs admits 
a labeling scheme of size $\frac{\ce{\Delta/2}+1}{2}\log n + O(\log \log n)$, with encoding
complexity $O(n \log n)$ and decoding complexity $O(\log n\log n)$. The average label size
of the latter scheme is $(1+o(1))\log n + O(\log\log n)$.
\end{theorem}

\begin{proof}

We use a simplified version of the method used in Section~\ref{sec:outerplanar}, thus many details
are omitted.
Our labeling scheme for planar graphs works with an embedding $\psi:V \rightarrow W_n$ of the given graph $G$
into $PL_n$. Levels of vertices $w\in W_n$ and $u\in V$ are defined as before.
Analogously to our previous scheme, we assign unique identifiers $Id(v)$ to vertices $v\in W_n$ level by level, 
starting from the highest level, and use them to define edge identifiers. Let $W^t_n \subset W_n$ 
denote the set of vertices in $PL_n$ with level at most $t$. Then we have the following lemma.

\begin{lemma}\label{lem:simplebounds_planar}
The following properties hold for every level $t$ and every vertex $v$ in level $t$.
\begin{enumerate}[i.]
	\item $|W_n^t| = O(\sqrt{N2^t})$.
	\item $|\Gamma(v)| = O(\sqrt{N2^{-t}})$.
\end{enumerate}
\end{lemma}

\begin{proof}
\textit{i.} Note that the number of vertices in level $i$ is $O(2^i\sqrt{\frac{N}{2^i}})$, thus for some constant $d$,
$$
|W_n^t| = c' \sum_{i=1}^t 2^{i-1}\sqrt{\frac{N}{2^i}} \leq c' \sqrt{N} \sum_{i=1}^t \sqrt{2}^{i-1} = O(\sqrt{N2^t}).
$$
\textit{ii.} Since $v$ is connected to vertices in clusters at distance at most $O(\log \Delta) = O(1)$ from its
cluster, there are $O(1)$ such clusters. Furthermore, the highest level of such a cluster is at least $t - O(1)$, thus
every such cluster contains $O(\delta_t) = O(\sqrt{N2^{-t}})$ vertices, which concludes the proof. $\qed$
\end{proof}

Lemma~\ref{lem:simplebounds_planar} implies that a the encoding length of label containing 
a vertex identifier $Id(\psi(v))$ together with $s \in \mathbb{N}$ edge identifiers corresponding
to $\psi(v)$ is 
\begin{equation}\label{eq:bound}
\log (\sqrt{N2^t}) + s \log (\sqrt{N2^{-t}}) + O(1) = \log n + \frac{s-1}{2}(\log n - t(v)) + O(1).
\end{equation}
Additionally storing the level $t(v)$ using $O(\log \log n)$ bits allows to perform decoding easily, without
needing the complications arising in our scheme for outerplanar graphs. Finally, combined with
the label splitting technique of Kannan et~al.~\cite{kannan1988implicit}
the statement of the theorem follows.

To prove the bound on the average label length we perform the following simple computation.
Consider a constant $B \gg \Delta$ and let $t^* = (1-B^{-1})\log n$. By choice of
$t^*$ and using Lemma~\ref{lem:simplebounds_planar} it follows that $|W_n^{t^*}| = c'n^{1-\epsilon}$ and 
$|W_n \setminus W_n^{t^*}| = n - c' n^{1-\epsilon}$ for some constants $c', \epsilon > 0$. Now, the length of every 
label $\la(v)$ for $v\in W_n^{t^*}$ can be trivially bounded by $\Delta \log n$, while for vertices $v\in W\setminus W^{t^*}_n$,
whose level is at least $t^*$, we use~(\ref{eq:bound}) to upper-bound the label size by 
$$
\log n + \frac{\Delta-1}{2}(\log n - t^*) + O(\log \log n) = \left(1 + B^{-1}\frac{\Delta-1}{2}\right) \log n + O(\log\log n).
$$
Finally, combining the bounds one obtains the following bound on the sum of the label sizes.
$$
\sum_{v\in V} |\la(v)| \leq c' n^{1-\epsilon} \Delta \log n + (n - c'n^{1-\epsilon})\left(1 + B^{-1}\frac{\Delta-1}{2}\right) \log n + O(n\log \log n).
$$
The term $d n^{1-\epsilon} \Delta \log n$ is clearly negligible, since $\Delta$ is fixed. We conclude that
the average label length is at most $\left(1 + B^{-1}\frac{\Delta-1}{2}\right)(1+ h(n)) \log n + O(\log \log n)$ for some function $h(n) = o(1)$. 
Repeating the argument with $B = O(\log \log n)$ leads to a bound of 
$$
(1+o(1))\log n + O(\log\log n)
$$ 
on the average label size, which is asymptotically the best possible average label size.

\end{proof}

In light of the $2\log n + O(\log\log n)$ labeling scheme of Gavoille and Labourel~\cite{gavoille2007shorter} for general
planar graphs, our result improves the best known bounds for $\Delta \leq 4$, and matches the bound in~\cite{gavoille2007shorter} 
for $\Delta \leq 6$.

\end{document}